\documentclass[11pt]{article}
\usepackage[margin=1in]{geometry}
\usepackage{amsmath,amssymb}
\usepackage{booktabs}
\usepackage{algorithmic}
\usepackage{hyperref}
\usepackage{enumitem}
\usepackage{float}
\usepackage{graphicx}
\usepackage{subcaption}
\usepackage{amsthm}
\usepackage{hyperref}
\usepackage{tikz}
\usetikzlibrary{shapes.geometric}

\definecolor{fiberblue}{RGB}{70,110,190}
\definecolor{fiberred}{RGB}{190,70,70}

\newcommand{\func}[1]{\textsc{#1}}
\newtheorem{lemma}{Lemma}[section]

\title{Fiber-Navigable Search: A Geometric Approach to Filtered ANN}
\author{Thuong Dang\thanks{Contact: dangt@uni-duesseldorf.de. Open to collaboration on benchmark implementation and large-scale evaluation.}}
\date{March 2026}

\begin{document}
\maketitle

\begin{abstract}
We present a geometric framework for filtered approximate nearest neighbor (ANN) search. Filtering a proximity graph by a metadata predicate produces a subgraph---a \emph{fiber}---whose connectivity and geometry can differ sharply from the full graph. Using local signals, we propose a two-phase search algorithm that combines full-graph exploration with filtered-neighbor descent when the local geometry is favorable. These signals also classify search failures into three regimes: topological cuts, geometric folds, and genuine basins. A key observation is that all three share a common resolution: restarting the search in a fiber-present cluster near the query. To support this, we introduce a lightweight anchor structure that identifies such regions and restarts the search accordingly. We show empirically that the method outperforms FAISS HNSW on filtered search and the three failure regimes separate cleanly and shift predictably with filter selectivity.
\end{abstract}

\tableofcontents

\section{Introduction}\label{sec:intro}

Nearest neighbor search in high-dimensional spaces is central to modern
retrieval systems, including product recommendation, semantic search,
and retrieval-augmented generation (RAG). As datasets grow to millions
or billions of vectors, exact search becomes too expensive. Approximate
nearest neighbor (ANN) algorithms are therefore essential for large-scale
vector retrieval.

Among ANN approaches, graph-based methods dominate due to their strong
recall--latency tradeoffs. These methods build a proximity graph over
the dataset and answer queries by greedily walking through the graph.
The two most influential methods are HNSW~\cite{malkov2020hnsw} and
Vamana (the graph construction algorithm behind
DiskANN~\cite{subramanya2019diskann}). HNSW builds a hierarchical
small-world graph with logarithmic search time. Vamana builds a flat
graph with bounded out-degree and is designed for disk-resident datasets.
Both achieve high recall on standard benchmarks and are widely used in
production.

In practice, however, queries rarely come without constraints.
A user searching for similar products may require a specific category,
price range, or brand. A document retrieval system may restrict results
to a particular language or date range. This setting is known as
\emph{filtered approximate nearest neighbor search}: given a query vector
$q$ and a predicate $S$ on metadata, find the approximate nearest
neighbors of $q$ among only points whose metadata satisfies $S$.

The core difficulty is structural: filtering removes nodes and edges
from the graph. A graph that is easy to navigate in full may become
hard to navigate once filtered. Edges needed for greedy descent may
vanish, leaving disconnected regions and local minima that trap the
search.

Existing systems use one of three strategies.
\emph{Post-filtering} retrieves the top-$N$ unfiltered results and
discards non-matching points---this fails when the true filtered
neighbors are far from the unfiltered ones.
\emph{Pre-filtering} restricts search to matching points, but the
resulting subgraph can be disconnected, trapping the search at local
minima.
\emph{Traversal-time filtering} walks the full graph but only collects
matching results; under selective filters, the walk converges in a
region shaped by the full graph, which may be far from the true
filtered neighbors.

All three strategies treat filtering as an operational constraint---they
decide \emph{when} to apply the predicate but rely on standard greedy
search for navigation. None accounts for how filtering changes the
geometry of the search space. As a result, algorithm design is largely
heuristic, and it remains unclear why certain queries fail or how
failures relate to filter selectivity.

In this paper, we propose a geometric approach. We view the filtered subset
of the graph as a \emph{fiber} over the metadata predicate and introduce
two local signals---\emph{fiber density} (the fraction of neighbors that pass the filter), and \emph{drift} (the average tendency of the filtered neighborhood toward or away from the query)---that describe its geometry during traversal. These signals drive a search algorithm that switches between \emph{fiber descent} (greedy search restricted to filtered neighbors) and full-graph exploration. They also classify search failures into three regimes:
\emph{topological cuts} (the fiber is locally disconnected),
\emph{geometric folds} (the fiber is present but slopes away from the query), and \emph{genuine basins} (true local minima on the fiber).
A key finding is that all three regimes share a common fix: restarting the search in a fiber-present region near the query. This leads to a simple algorithm that achieves strong filtered recall without building filter-aware graph construction.

\subsection{Contributions}
 
This paper makes the following contributions:
 
\begin{enumerate}
 
\item \textbf{A geometric framework and failure taxonomy for filtered
graph search.}
We formalize filtered ANN search using a fibered view of the dataset
and introduce two local signals---\emph{fiber density}, and \emph{drift}---that describe the filtered region's geometry
during traversal. These signals classify failures at \emph{stall points}
(nodes where the walk stops without finding enough results) into three
regimes: \emph{topological cuts}, \emph{geometric folds}, and
\emph{genuine basins}. We introduce the \emph{boundary-improving set}
to diagnose what happens outside the fiber at stall points, and show
empirically that the three regimes separate cleanly on every diagnostic
axis and shift predictably with filter selectivity.
 
\item \textbf{A drift-guided two-phase search algorithm.}
The drift signal drives a two-phase traversal: the algorithm begins
with fiber descent (greedy search restricted to filtered neighbors)
when drift is negative, and falls back to full-graph beam exploration
when drift turns non-negative. Phase switching is dynamic---the
algorithm re-enters fiber descent whenever the drift becomes favorable
again. This separation between fiber navigation and full-graph
navigation avoids the need for filter-aware index construction.
 
\item \textbf{An anchor-based restart mechanism.}
When a walk stalls, all three failure regimes share a common fix:
restarting in a fiber-present region near the query. To support this,
we introduce a lightweight \emph{anchor atlas}---a clustering-based
structure that identifies fiber-present regions near the query in
$O(|S|)$ time using per-cluster metadata statistics and an inverted
cluster index.
 
\item \textbf{Experimental validation.}
On a real-world dataset of 105{,}100 vectors with 24 metadata fields,
our method outperforms FAISS HNSW under both post-filtering and
traversal-time filtering, with near-zero failure rate across
10{,}000 test queries spanning filter selectivities from
${<}0.01\%$ to ${\sim}21\%$.
 
\end{enumerate}

\paragraph{Paper outline.}
Section~\ref{sec:related_work} reviews prior work on filtered approximate nearest neighbor search.
Section~\ref{sec:framework} introduces the geometric framework and the local signals used to characterize filtered search.
Section~\ref{sec:construction} describes index construction and the anchor atlas.
Section~\ref{sec:algorithms} presents the search algorithms.
Sections~\ref{sec:experiments} and~\ref{sec:results} describe the experimental setup and report the results.
Section~\ref{sec:analysis} analyzes stall regimes empirically.
Section~\ref{sec:conclusion} concludes and discusses future directions.

\section{Related Work}\label{sec:related_work}

Filtered approximate nearest neighbor search (FANNS) has attracted significant attention in recent years. Existing approaches can be broadly categorized based on how they integrate filtering with the underlying vector index.

\textbf{Filter-aware index construction.}
Several systems incorporate filter information directly into the index structure.
Filtered-DiskANN~\cite{gollapudi2023filtereddiskann} introduces two algorithms that modify the Vamana graph construction to account for label constraints, improving navigability of filter-induced subgraphs.
UNG~\cite{cai2024ung} further integrates label-set relationships with vector proximity through a unified navigating graph.
Earlier hybrid query systems such as HQANN~\cite{wu2022hqann} and NHQ~\cite{wang2023nhq} combine vector similarity with attribute information during index construction.

\textbf{Predicate-agnostic graph methods.}
Instead of modifying the graph for each filter, predicate-agnostic methods attempt to maintain robustness under filtering.
ACORN~\cite{patel2024acorn} expands neighbor lists during HNSW construction to preserve connectivity when nodes are removed by filters.

\textbf{Index-agnostic search strategies.}
Other approaches leave the index unchanged and adapt the search procedure.
RWalks~\cite{aitoamar2025rwalks} diffuses attribute information across the graph via random walks, enabling efficient filtered search without modifying the index structure.
AIRSHIP~\cite{zhao2022airship} performs constrained search on proximity graphs by integrating UDF filtering directly into the search procedure, and proposes three optimizations: starting point selection, multi-direction search, and biased priority queue selection, all without modifying the underlying graph. Our approach shares this search-time perspective but differs in that it uses explicit local signals to guide traversal and a structured, metadata-aware restart mechanism.

\textbf{Range-filtered ANN.}
Some work focuses specifically on range predicates.
SeRF~\cite{zuo2024serf} builds segment graphs for range filtering, while WoW~\cite{wang2025wow} introduces incremental window-based indexing.

\textbf{Benchmarks and surveys.}
Recent studies emphasize the need for systematic evaluation of filtered ANN methods.
Iff et al.~\cite{iff2025benchmark} provide a large-scale benchmark of filtered vector search algorithms, and Lin et al.~\cite{lin2025survey} present a comprehensive survey of FANNS techniques.

\section{Geometric Framework}\label{sec:framework}

We note that the geometric perspective is not required to implement the algorithm; all quantities reduce to simple graph statistics. However, it provides a useful lens for understanding why filtered search fails and directly motivates the design of the algorithmic components.

\subsection{Problem Setting}\label{sec:setting}

\paragraph{Data.}
Let $X = \{x_1,\dots,x_n\} \subset \mathbb{R}^d$. Each point carries
metadata across $F$ fields, with field $f$ taking values in a discrete
set $V_f$. The metadata map $m: X \to \prod_f V_f$ sends each point to
its metadata tuple. A \emph{filter predicate} specifies, for a subset of
fields, a set of allowed values per field:
$S = \{f_{i_1} \in A_1,\dots, f_{i_s} \in A_s\}$ where
$A_j \subseteq V_{f_{i_j}}$. A point satisfies $S$ when
$f_{i_j}(x) \in A_j$ for all $j = 1,\dots,s$; unconstrained fields are unrestricted. 

\paragraph{Proximity graph.}
A global graph $G = (X, E)$ (e.g.\ HNSW, Vamana, $k$NN-graph) provides
an adjacency list $N(x)$ for every point.

\subsection{Fibers over Metadata}

Let $m : X \to \mathcal{M}$ map the dataset into the product metadata space as described in Section~\ref{sec:setting}. For a filter predicate $S$, the \emph{fiber} is
\[
X_S = \{ x \in X \mid x \text{ satisfies } S \}.
\]

The \emph{fiber subgraph} $G_S = (X_S,\, E \cap (X_S \times X_S))$ is the subgraph of $G$ induced by $X_S$. Standard greedy or beam search on $G$ performs discrete descent on the potential $V(x) = d(q,x)$. However, navigability of $G$ does not imply navigability of $G_S$: the fiber subgraph may lose edges critical for descent toward $q$.

We model the dataset as a discrete fibered space:
\begin{itemize}[nosep]
  \item The \emph{base space} is the metadata space $\mathcal{M}$.
  \item The \emph{fiber} above each metadata value $u \in \mathcal{M}$ is the set of
        points $\{x : m(x) = u\}$ together with their embedding coordinates.
  \item The filter $S$ selects a union of such fibers.
\end{itemize}

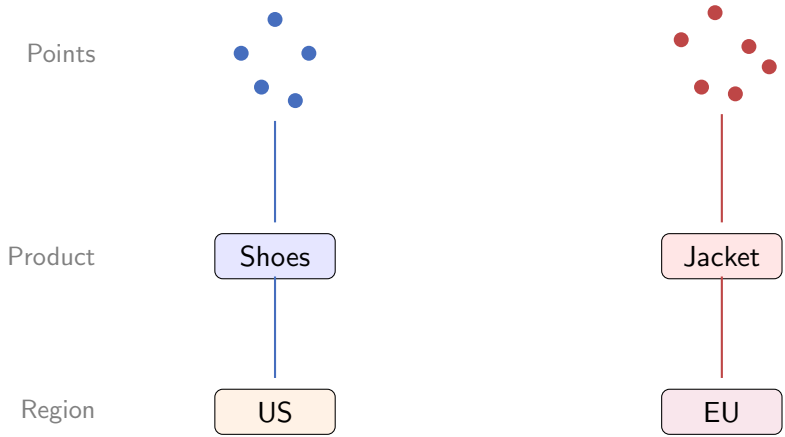
\begin{figure}[ht]
\centering
\begin{tikzpicture}[scale=0.9]


  \node[circle, inner sep=2pt, fill=fiberblue] (p1a) at (2.5, 10.0) {};
  \node[circle, inner sep=2pt, fill=fiberblue] (p1b) at (3.0, 10.5) {};
  \node[circle, inner sep=2pt, fill=fiberblue] (p1c) at (3.5, 10.0) {};
  \node[circle, inner sep=2pt, fill=fiberblue] (p1d) at (2.8, 9.5) {};
  \node[circle, inner sep=2pt, fill=fiberblue] (p1e) at (3.3, 9.3) {};

  \draw[fiberblue, thick] (3.0, 9.0) -- (3.0, 7.5);

  \node[draw, rounded corners=3pt, fill=blue!10,
        minimum width=1.6cm, minimum height=0.6cm,
        font=\sffamily] (A) at (3.0, 7.0) {Shoes};

  \draw[fiberblue, thick] (3.0, 6.7) -- (3.0, 5.2);

  \node[draw, rounded corners=3pt, fill=orange!10,
        minimum width=1.6cm, minimum height=0.6cm,
        font=\sffamily] (X) at (3.0, 4.7) {US};


  \node[circle, inner sep=2pt, fill=fiberred] (p2a) at (9.0, 10.2) {};
  \node[circle, inner sep=2pt, fill=fiberred] (p2b) at (9.5, 10.6) {};
  \node[circle, inner sep=2pt, fill=fiberred] (p2c) at (10.0, 10.1) {};
  \node[circle, inner sep=2pt, fill=fiberred] (p2d) at (9.3, 9.5) {};
  \node[circle, inner sep=2pt, fill=fiberred] (p2e) at (9.8, 9.4) {};
  \node[circle, inner sep=2pt, fill=fiberred] (p2f) at (10.3, 9.8) {};

  \draw[fiberred, thick] (9.6, 9.1) -- (9.6, 7.5);

  \node[draw, rounded corners=3pt, fill=red!10,
        minimum width=1.6cm, minimum height=0.6cm,
        font=\sffamily] (B) at (9.6, 7.0) {Jacket};

  \draw[fiberred, thick] (9.6, 6.7) -- (9.6, 5.2);

  \node[draw, rounded corners=3pt, fill=purple!10,
        minimum width=1.6cm, minimum height=0.6cm,
        font=\sffamily] (Y) at (9.6, 4.7) {EU};

  \node[font=\small\sffamily, gray, anchor=east]
    at (0.5, 10.0) {Points};
  \node[font=\small\sffamily, gray, anchor=east]
    at (0.5, 7.0) {Product};
  \node[font=\small\sffamily, gray, anchor=east]
    at (0.5, 4.7) {Region};

\end{tikzpicture}
\caption{%
  Fibers over metadata.  %
}
\label{fig:fibered-space}
\end{figure}

In the figure above, points in the embedding space (top) project onto product values and then onto region values. Each vertical thread is a fiber: the left fiber $m^{-1}(\text{``US''},\text{``Shoes''})$ consists of all points with product${=}$ Shoes and region${=}$ US; the right fiber $m^{-1}(\text{``EU''}, \text{``Jacket''})$ consists of all points with product${=}$ Jacket and region${=}$ EU.  A filter selects one or more fibers.

\subsection{Local Signals}\label{sec:signals}

Every graph search performs discrete gradient descent on the potential $V(x) = d(q, x)$. Under filter $S$, the \textit{filtered neighborhood} is $N_S(x) = \{y \in N(x) \mid y \text{ satisfies } S\} = N(x) \cap X_S$. We next define two local signals that characterize the local geometry of the fiber at a node $x$:

\paragraph{Fiber density.} $\rho_S(x) = |N_S(x)| / |N(x)|$. It is a topological signal measuring how much of the local connectivity survives filtering.
Low density indicates that the fiber is locally sparse or disconnected, while high density indicates that the fiber structure is intact and navigable.

\paragraph{Drift.}
When $|N_S(x)| > 0$, we define
\[
\text{drift}(x) = \frac{1}{|N_S(x)|} \sum_{y \in N_S(x)} \bigl(V(y) - V(x)\bigr).
\] Drift measures the \textit{average tendency of the neighborhood}—whether the region collectively slopes toward the query or not. It is the discrete Laplacian of $V$ restricted to the fiber. Negative drift indicates a wide descent valley (Figure \ref{fig:negative-drift}); non-negative drift might indicate a narrow ridge (Figure \ref{fig:positive-drift}).

\begin{figure}[H]
\centering
\begin{subfigure}[b]{0.48\textwidth}
    \centering
    \includegraphics[width=0.8\textwidth]{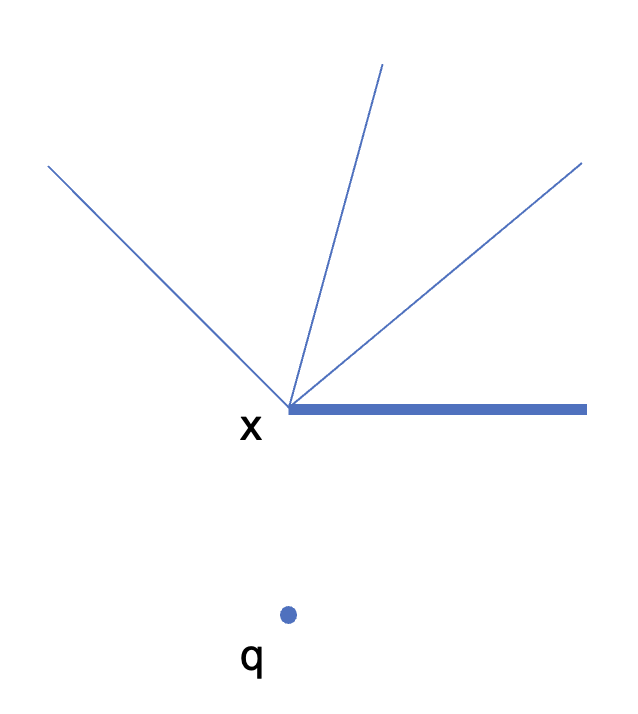}
    \caption{$\operatorname{drift}(x) > 0$: stall region}
    \label{fig:positive-drift}
\end{subfigure}
\hfill
\begin{subfigure}[b]{0.48\textwidth}
    \centering
    \includegraphics[width=0.8\textwidth]{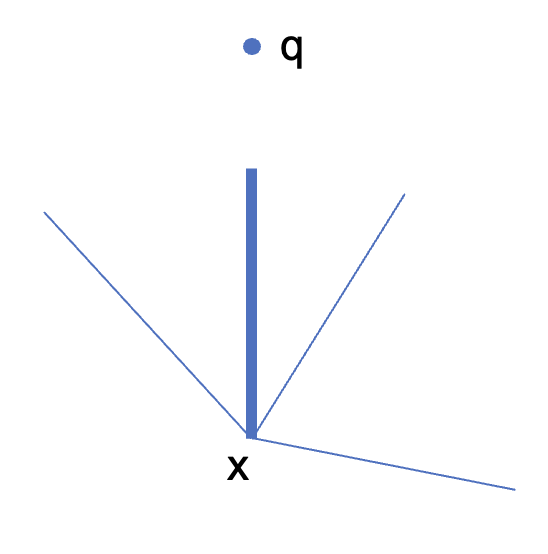}
    \caption{$\operatorname{drift}(x) < 0$: wide region for descent.}
    \label{fig:negative-drift}
\end{subfigure}
\caption{Drift and potential best descent direction (bold lines) as filtered neighborhood diagnostics.}
\label{fig:drift-signals}
\end{figure}


\subsection{Algorithm Overview}\label{sec:overview}

We outline the search algorithm at a high level so
that the reader has a concrete picture to anchor the subsequent analysis. 
The algorithm admits two instantiations, corresponding to different levels of engagement
with the fiber geometry:

\paragraph{Beam search (cluster-level navigation).}
The simpler variant uses only the atlas structure. The anchor atlas selects a fiber-present cluster near $q$, a standard beam search walks the full graph from the seed while passively
collecting filtered points, and when the walk converges the atlas transitions to the next cluster. This version navigates \emph{between} clusters but is geometry-unaware \emph{within} a cluster: it does not use drift during the walk.

\paragraph{Drift-guided search (cluster- and fiber-level navigation).}
The full variant adds intra-cluster geometric navigation via the drift signal. From each restart, the walk alternates between fiber descent (Phase~1) and full-graph beam search
(Phase~2), re-entering Phase~1 whenever the drift turns favorable. This version navigates both between and within clusters, exploiting the fiber structure at every expansion step.

\medskip
\noindent Both variants share the outer restart loop and the anchor atlas. The guided
variant achieves higher recall with a smaller beam by spending its expansion budget on
drift-identified productive regions rather than blind beam exploration.

Given a query $q$ and a metadata filter $S$, the algorithm proceeds as follows:
\begin{enumerate}
    \item \textbf{Cluster selection.} The anchor atlas---a set of $k$-means clusters augmented
    with per-cluster metadata statistics---identifies clusters that contain points matching $S$
    and lie close to $q$. Seed points are drawn from the top-scoring clusters.
    \item \textbf{Local walk.} From these seeds, a walk explores the full proximity graph $G$
    (which is always navigable), collecting any encountered points that match $S$. Two walk
    strategies are available:
    \begin{itemize}
        \item \emph{Beam search}: a standard beam search on $G$ that passively collects
        filtered points along the way.
        \item \emph{Drift-guided search}: a two-phase walk that descends on filtered neighbors when the drift signal is favorable (Phase~1) (Figure \ref{fig:local_descent}), queuing the top $K_f$ filtered neighbors that are closer to the query at each step, and falls back to full-graph beam search when it is not (Phase~2). The phases are not a one-way progression: when Phase~2 traversal encounters a region where drift turns negative and filtered neighbors reappear, the algorithm re-enters Phase~1 and resumes fiber descent.
    \end{itemize}
    \item \textbf{Restart.} If a walk terminates without collecting enough results, the atlas
    provides seeds from the next-best cluster and a new walk begins. A small jump budget
    (typically 3--4 restarts) suffices.
\end{enumerate}

\begin{figure}[H]
\centering
\includegraphics[width=0.35\textwidth]{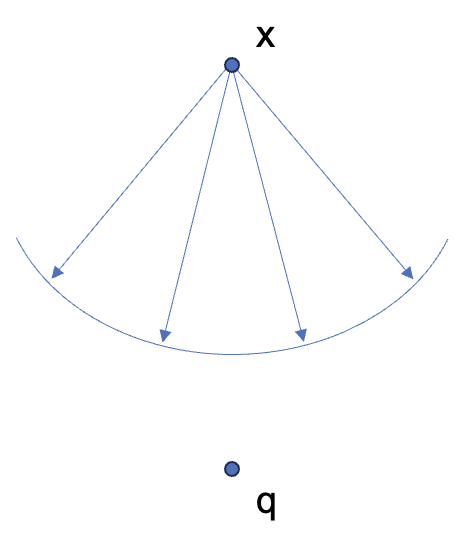}
\caption{Local descent.}
\label{fig:local_descent}
\end{figure}

\noindent The two local signals defined in Section~\ref{sec:signals}---fiber density, and drift---drive every decision in this loop: drift governs phase switching within a walk, fiber density characterizes stall points, and both inform the restart mechanism. The next subsection develop the geometric perspective that motivates these choices.

\subsection{Geometric Motivation}\label{sec:motivation}

The algorithm design is guided by an analogy with differential geometry (see e.g., \cite{lee2013smooth}).
A \textit{smooth manifold} is usually too complex to describe with a single coordinate system, so one covers it with an \emph{atlas} of overlapping \emph{charts}, each providing local coordinates in which calculus is tractable. When a chart’s coordinates break down—such as near the boundary of its domain—one transitions to an overlapping chart via a \emph{transition function} and continues (Figure \ref{fig:charts}). A \emph{fiber bundle} $\pi: E \to B$ adds a further layer: over each point $x \in B$, the fiber $\pi^{-1}(x)$ is diffeomorphic to a fixed space $F$, called the \emph{typical fiber} (Figure \ref{fig:fibers}). The total space $E$ is locally (but not necessarily globally) a product $B \times F$. Navigation on the bundle requires both moving along the base and keeping track of how the fibers are attached.

Our filtered search problem instantiates this picture on a discrete, finite structure.
Table~\ref{tab:correspondence} records the dictionary.

\begin{table}[h]
\centering
\caption{Correspondence between differential geometry and filtered graph search.}
\label{tab:correspondence}
\begin{tabular}{@{}lll@{}}
\toprule
\textbf{Differential geometry} & \textbf{Filtered graph search} & \textbf{Role} \\
\midrule
Chart & Cluster (anchor region) & Local region with tractable geometry \\
Atlas & Anchor atlas & Collection of charts covering the space \\
Transition function & Anchor restart & Map from one chart to another \\
Fiber over a point & $m^{-1}(u)$ for $u\in \mathcal{M}$ & Points with metadata value $u$ \\
Fiber connectivity & Fiber density & How much of the fiber is present locally \\
Laplacian of $V$ on fiber & $\mathrm{drift}(x)$ & Whether the region slopes toward the query \\
\bottomrule
\end{tabular}
\end{table}

\begin{figure}[H]
\centering
\includegraphics[width=0.8\textwidth]{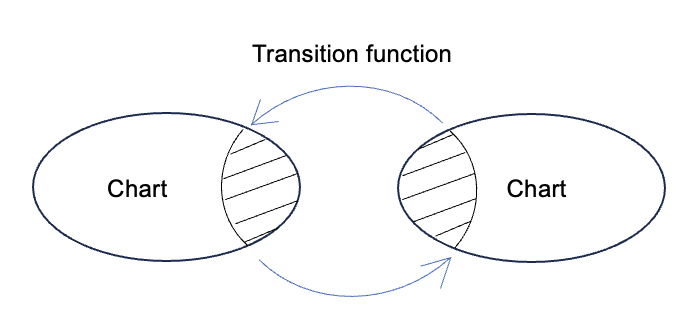}
\caption{Transition function to glue charts in differential geometry.}
\label{fig:charts}
\end{figure}

\begin{figure}[H]
\centering
\begin{subfigure}[b]{0.48\textwidth}
    \centering
    \includegraphics[width=\textwidth]{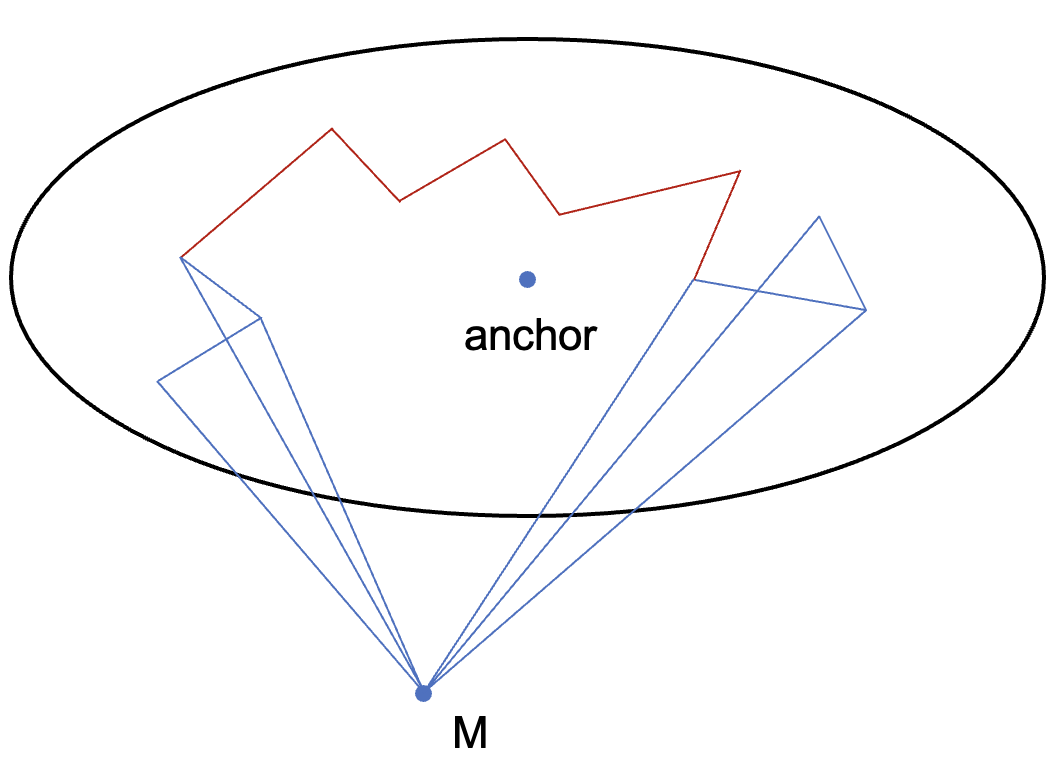}
    \caption{Fiber over metadata inside a cluster}
    \label{fig:discrete-fiber}
\end{subfigure}
\hfill
\begin{subfigure}[b]{0.48\textwidth}
    \centering
    \includegraphics[width=\textwidth]{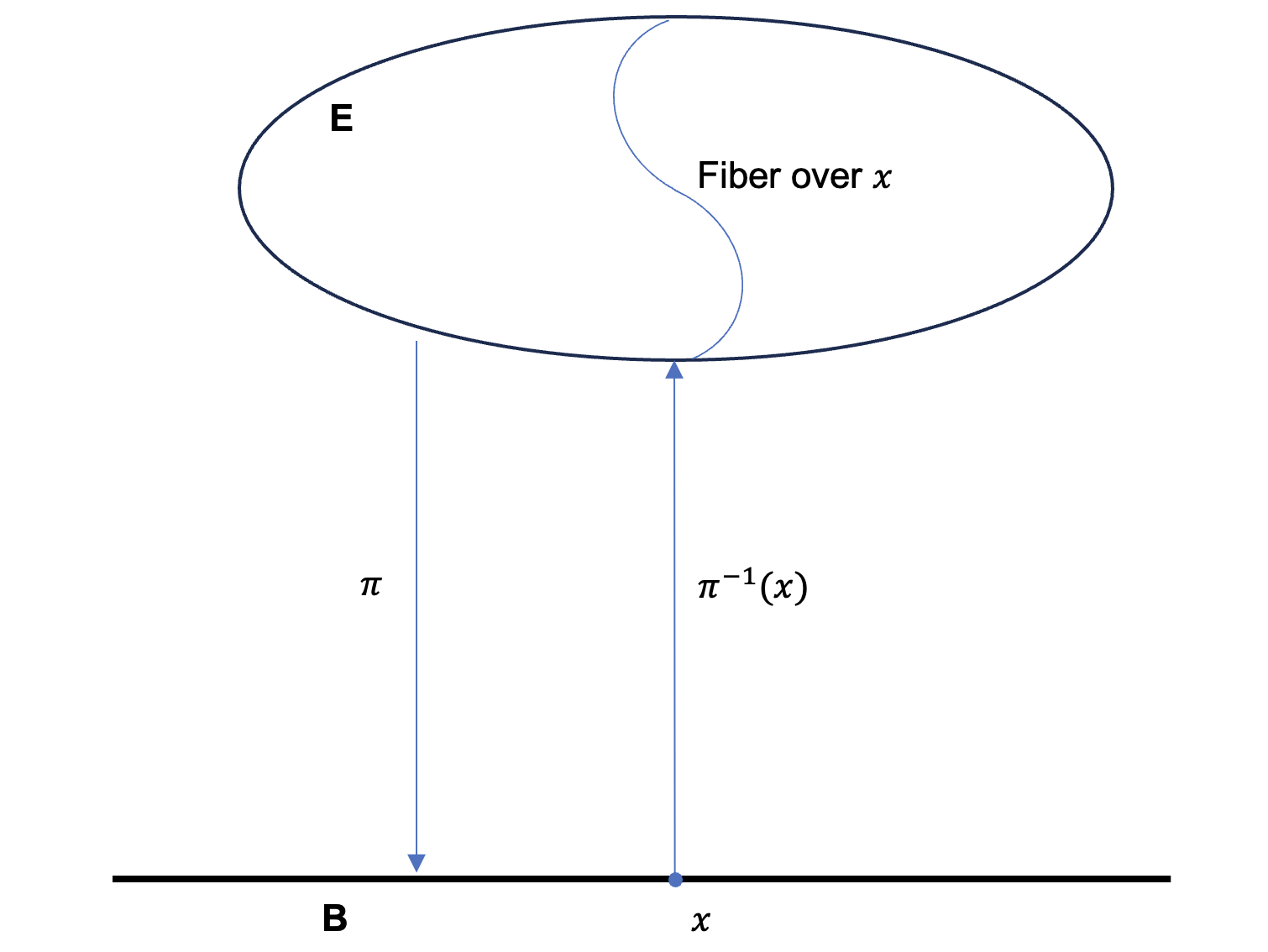}
    \caption{Fiber over a point $x$ in differential geometry}
    \label{fig:geometry-fiber}
\end{subfigure}
\caption{Discrete and continuous fiber structures.}
\label{fig:fibers}
\end{figure}

\paragraph{Fiber connectivity.}
In a smooth fiber bundle $\pi: E \to B$, every fiber is homeomorphic to the typical fiber $F$---so if $F$ is connected, all fibers are.
In our discrete setting, the situation is worse: filtering can
disconnect or entirely eliminate the fiber locally.

\paragraph{Drift as the discrete Laplacian.}
In $\mathbb{R}^n$, Laplacian of a function is the sum of its second derivatives. In general, it measures how different a point is from the average of its neighbors. Drift is the discrete analogue restricted to the fiber: negative drift means the filtered neighborhood is on average closer to the query, indicating a favorable region for descent.
\section{Index Construction}\label{sec:construction}
 
The index consists of a proximity graph and an anchor atlas. The search algorithm is graph-agnostic: it operates on any graph exposing a \func{GetNeighbors} interface. The anchor atlas requires no filter-aware modifications.
 
\subsection{Proximity Graph}
 
We introduce the \textbf{$\alpha$-$k$NN graph}: a symmetrized $k$-nearest-neighbor graph with selective $\alpha$-RNG (relative neighborhood graph) degree capping---a simplified version of DiskANN \cite{subramanya2019diskann}. Construction proceeds in three stages:
 
\medskip
\noindent\textbf{Algorithm 1: $\alpha$-$k$NN Graph Construction}
\begin{algorithmic}[1]
\REQUIRE Vectors $\{\hat{x}_1, \dots, \hat{x}_n\}$ (normalized), degree $k$, max degree $R_{\max}$, relaxation $\alpha > 1$
\ENSURE Adjacency lists $N(x)$ for all $x$
 
\medskip
\STATE \textbf{Stage 1: Directed $k$NN}
\STATE Compute $k$ nearest neighbors for each point using cosine similarity
\STATE $N(x) \leftarrow \text{kNN}(x)$ for all $x$ \COMMENT{directed, degree exactly $k$}
 
\medskip
\STATE \textbf{Stage 2: Symmetrization}
\FOR{each edge $(i, j) \in G$}
    \STATE Add reverse edge $(j, i)$ if not present
\ENDFOR
\COMMENT{Degree increases; hub nodes may have degree $\gg k$}
 
\medskip
\STATE \textbf{Stage 3: Selective $\alpha$-RNG Pruning}
\FOR{each node $i$ with $|N(i)| > R_{\max}$}
    \STATE Sort $N(i)$ by distance to $i$
    \STATE $\text{kept} \leftarrow \emptyset$
    \FOR{each $p \in N(i)$ in distance order}
        \IF{$\forall\, q \in \text{kept}:\ d(i,p) < \alpha \cdot d(q,p)$}
            \STATE $\text{kept} \leftarrow \text{kept} \cup \{p\}$
        \ENDIF
        \IF{$|\text{kept}| \geq R_{\max}$} \STATE \textbf{break} \ENDIF
    \ENDFOR
    \STATE $N(i) \leftarrow \text{kept}$
\ENDFOR
\COMMENT{Nodes with $|N(i)| \leq R_{\max}$ are unchanged}
\end{algorithmic}
 
\paragraph{Design rationale.} Symmetrization ensures navigability: if $A$ can reach $B$, then $B$ can reach $A$, enabling the walk to explore in both directions. However, symmetrization creates hub nodes with degree far exceeding $k$ (up to $500\times$ on a 100k dataset). The selective $\alpha$-RNG pruning caps only over-degree nodes, preserving local connectivity for typical nodes while eliminating pathological hubs. The $\alpha > 1$ relaxation retains directionally diverse edges rather than keeping only the nearest, providing coverage across different approach directions---beneficial for Phase~1 fiber descent, which needs filtered neighbors in the descent direction toward $q$.
 
\paragraph{Graph-agnostic search.} The guided search algorithm (Section \ref{sec:algorithms}) depends only on the \func{GetNeighbors} interface, not on the graph construction method. To validate this, we also evaluate on the base layer (level~0) of a FAISS HNSW index, extracted without modification. Using the same HNSW graph that FAISS searches with its own hierarchical algorithm, our guided search substantially outperforms both FAISS's post-filter and traversal-filter strategies. This confirms that the improvement derives from the search algorithm and anchor atlas, not from the choice of base graph.

\subsection{Anchor Atlas}\label{sec:anchor}

The anchor atlas is a lightweight clustering-based structure that supports efficient seed selection for filtered queries. It consists of $K$ clusters (obtained via $k$-means on the embedding vectors), augmented with per-cluster metadata statistics and an inverted index for fast cluster pruning.

\paragraph{Storage.}
For each cluster $c$ and each metadata field $f$, the atlas stores:
\begin{itemize}
    \item \texttt{members}$[c][f][v]$: the list of point indices in $c$ with value $v$ in field $f$.
\end{itemize}
Storage is $O(nF)$ (Lemma \ref{lem:storage}), where $F$ is the number of metadata fields.

\paragraph{Scope note.} For clarity, we describe the atlas operations for \textit{single-value} conjunctive filters $S = \{f_1 = v_1, \ldots, f_m = v_m\}$; the extension to set-valued predicates $f_i \in A_i$ replaces each membership check with a set lookup over $A_i$.

\paragraph{Inverted cluster index.}
To avoid scanning all $K$ clusters at query time, the atlas maintains an inverted index:
$$
\texttt{cluster\_index}[f][v] = \{C \mid \exists x\in C, f(x) = v \},
$$
mapping each (field, value) pair to the set of clusters containing at least one matching point. For filter $S$, the candidate clusters are obtained by intersecting $m$ posting lists:
$$
C_{\text{match}} = \bigcap_{i=1}^{m} \texttt{cluster\_index}[f_i][v_i].
$$
This costs $O(|S|)$ regardless of $K$. For selective filters, $|C_{\text{match}}| \ll K$: a rare value like ``Chem.\ cosmetics'' matches 3--5 clusters out of 324 in our experiment (Section~\ref{sec:experiments}), reducing subsequent scoring cost by two orders of magnitude.

\paragraph{Anchor selection.}
At query time, given query $q$ and filter $S$, the atlas: (1)~retrieves $C_{\text{match}}$ via the inverted index; (2)~discards already-processed clusters; (3)~ranks remaining clusters by similarity score between $q$ and cluster centroid; and (4)~seeds are drawn until the seed budget is filled.

\subsection{Complexity and Scaling}\label{sec:scale}

\paragraph{Index construction cost.}
The proximity graph dominates indexing time: $k$-NN computation costs $O(n^2 \cdot d)$ brute-force (reducible with approximate methods), symmetrization is $O(|E|)$, and selective $\alpha$-RNG pruning is $O(n \cdot R_{\max} \log R_{\max})$. The anchor atlas is cheap by comparison: $k$-means runs in $O(n \cdot K \cdot d \cdot T_{\text{iter}})$, and building the per-cluster hash tables, \texttt{members} lists, and inverted index requires a single $O(n \cdot F)$ pass over the metadata, where $F$ is the number of metadata fields.

\paragraph{Space.}

\begin{lemma}[Atlas storage]\label{lem:storage}
Let each of the $n$ points be assigned to exactly one of $K$ clusters,
and let each point carry metadata across $F$ fields. The total storage for \texttt{members}, and \texttt{cluster\_index} is $O(n \cdot F)$, independent of the vocabulary sizes $|V_f|$ and the number of clusters $K$.
\end{lemma}

\begin{proof}
Each point $x_i$ belongs to one cluster $c(x_i)$ and has one value
$f(x_i)$ per field. For each field $f$, point $x_i$ contributes:
\begin{itemize}[nosep]
  \item one entry in the list
        $\texttt{members}[c(x_i)][f][f(x_i)]$,
  \item at most one insertion of $c(x_i)$ into
        $\texttt{cluster\_index}[f][f(x_i)]$
        (deduplicated).
\end{itemize}
Summing over $n$ points and $F$ fields, each structure contains at
most $n \cdot F$ non-zero entries. Only non-zero entries are stored.
\end{proof}

The graph stores $O(n \cdot \bar{R})$ edges. The anchor atlas stores \texttt{members} lists totaling $O(n \cdot F)$ entries. The inverted index tables is sparse: only non-zero entries are stored, so their size is also $O(n \cdot F)$ (Lemma \ref{lem:storage}). When metadata is sparse (not all fields are populated for every
point), only non-empty field values contribute, and the bound
tightens to $O(\sum_{i=1}^{n} F_i)$ where $F_i$ is the \textit{number of non-empty fields for point} $x_i$.

\paragraph{Scaling the anchor atlas.}
With $K = \lceil\sqrt{n}\rceil$ clusters, the brute-force cosine computation over $C_{\text{match}}$ is negligible for $n \leq 10^6$. When the dataset scales, there will be two viable options

\begin{itemize}
    \item We can build a \emph{hierarchical anchor structure}: two-level $k$-means with $K_1 = n^{1/4}$ super-clusters, each containing $K_2 = n^{1/4}$ sub-clusters. The inverted index operates at both levels: first identify matching super-clusters in $O(|S|)$, then search within them. Each level is small enough for brute-force cosine, avoiding the filtered graph search problem entirely while keeping the total anchor scoring cost at $O(n^{1/4} \cdot d)$ per query
    \item We can also fix $K$ to a moderate constant (e.g., $K \in [100, 1000]$) independent of $n$. The anchor atlas becomes a coarser but cheaper structure, and the graph walk compensates with local exploration. This is justified by two observations. First, the complexity analysis in Section~\ref{sec:algorithms} shows that walk cost dominates query time over anchor scoring, so even a brute-force scan over all $K$ clusters adds negligible overhead when $K$ is bounded. Second, the anchor atlas does not need to identify the \emph{optimal} fiber-present region---it needs to land the walk in a \emph{sufficiently good} region where the fiber exists and the graph can take over.
\end{itemize}

\section{Search Algorithms}\label{sec:algorithms}

Both search strategies share the outer restart loop (Algorithm~2). They differ in how a single walk is executed.

\medskip
\noindent\textbf{Algorithm 2: Outer Search Loop with Anchor Restarts}
\begin{algorithmic}[1]
\REQUIRE Query $q$, filter $S$, result size $k$, jump budget $J$, walk procedure $\func{Walk}$, seed budget $n_s$, cluster budget $C_{\max}$
\ENSURE Approximate $k$-NN of $q$ in $X_S$

\STATE $\text{results} \leftarrow \emptyset$, \quad $\text{processed} \leftarrow \emptyset$
\FOR{$j = 0, 1, \dots, J$}
\STATE \textbf{Anchor selection:}
\STATE \quad Retrieve candidate clusters $C_{\text{match}}$ via inverted index
\STATE \quad Discard clusters in $\text{processed}$
\STATE \quad Rank remaining by similarity score
\STATE \quad $\text{seeds} \leftarrow \emptyset$
\FOR{each cluster $c$ in top-$C_{\max}$ ranked clusters}
    \IF{$|\text{seeds}| \geq n_s$}
        \STATE \textbf{break}
    \ENDIF
    \STATE Sample filter-matching points from $c$ via \texttt{members} intersection
    \STATE Add sampled points to $\text{seeds}$; add $c$ to $\text{processed}$
\ENDFOR
\IF{no seeds found}
\STATE \textbf{break}
\ENDIF

\STATE $\text{walk\_results} \leftarrow \func{Walk}(q, \text{seeds}, S)$
\STATE $\text{results} \leftarrow \text{results} \cup \text{walk\_results}$ \COMMENT{deduplicate, keep best similarity}

\IF{$|\text{results}| \geq k$}
\STATE \textbf{break} \COMMENT{enough results collected}
\ENDIF
\ENDFOR
\RETURN top-$k$ from results by similarity
\end{algorithmic}

\subsection{Walk Strategy 1: Beam Search}\label{sec:beam}

\medskip
\noindent\textbf{Algorithm 3: Beam Search Walk (Passive Filtered Collection)}
\begin{algorithmic}[1]
\REQUIRE Query $q$, seeds, filter conditions $S$, beam width $B$
\ENSURE Filtered results collected along walk

\STATE $\text{candidates} \leftarrow \{(\cos(q, s),\; s) : s \in \text{seeds}\}$, sorted descending
\STATE $\text{candidates} \leftarrow \text{candidates}[{:}B]$ \COMMENT{keep top $B$}
\STATE $\text{expanded} \leftarrow \emptyset$, \quad $\text{results} \leftarrow \emptyset$

\WHILE{$\exists$ unexpanded node in candidates}
    \STATE $x \leftarrow$ best unexpanded node in candidates
    \STATE $\text{expanded} \leftarrow \text{expanded} \cup \{x\}$

    \FOR{each $y \in N(x)$ not yet seen}
        \STATE Compute $\cos(q, y)$
        \STATE Add $({\cos(q,y)},\; y)$ to candidates
        \IF{$y$ matches filter $S$}
            \STATE $\text{results} \leftarrow \text{results} \cup \{(\cos(q,y),\; y)\}$
        \ENDIF
    \ENDFOR

    \STATE Sort candidates descending; keep top $B$ \COMMENT{prune beam}
\ENDWHILE

\RETURN results
\end{algorithmic}

\subsection{Walk Strategy 2: Drift-Guided Search}\label{sec:guided}

\medskip
\noindent\textbf{Algorithm 4: Drift-Guided Walk (Two-Phase Navigation)}
\begin{algorithmic}[1]
\REQUIRE Query $q$, seeds, filter conditions $S$, beam width $B$, frontier width $K_f$, stall budget $T$
\ENSURE Filtered results collected along walk

\STATE Initialize potential cache: $V(s) \leftarrow 1 - \cos(q,s)$ for each seed $s$
\STATE $\text{frontier} \leftarrow \{(V(s), s) : s \in \text{seeds}\}$ sorted ascending \COMMENT{Phase 1 queue}
\STATE $\text{beam} \leftarrow \emptyset$ \COMMENT{Phase 2 queue}
\STATE $\text{phase} \leftarrow 1$, \quad $\text{stall} \leftarrow 0$, \quad $\text{results} \leftarrow \{\}$

\WHILE{expansion budget not exhausted}

    \STATE \textbf{--- Node Selection ---}
    \IF{phase $= 1$}
        \STATE $x \leftarrow$ pop best unexpanded from frontier
        \IF{frontier exhausted}
            \STATE $\text{phase} \leftarrow 2$; seed beam from all seen unexpanded nodes; \textbf{continue}
        \ENDIF
    \ELSE
        \STATE $x \leftarrow$ pop best unexpanded from beam
        \IF{beam exhausted} \STATE \textbf{break} \COMMENT{walk converged} \ENDIF
        \IF{$V(x) > V_{(k)}$} \STATE \textbf{break} \COMMENT{cannot improve top-$k$} \ENDIF
        \IF{$\text{stall} \geq T$} \STATE \textbf{break} \COMMENT{fiber too sparse, defer to restart} \ENDIF
    \ENDIF

    \STATE \textbf{--- Expand $x$ ---}
    \FOR{each unseen $y \in N(x)$}
        \STATE $V(y) \leftarrow 1 - \cos(q,y)$; cache $V(y)$
        \IF{$y$ matches filter} \STATE $\text{results}[y] \leftarrow \cos(q,y)$ \ENDIF
    \ENDFOR

    \STATE \textbf{--- Fiber Diagnostics ---}
    \STATE $N_S(x) \leftarrow \{y \in N(x) : y \text{ matches filter}\}$
    \STATE $\rho_S(x) \leftarrow |N_S(x)| / |N(x)|$
    \STATE $\text{drift}(x) \leftarrow \frac{1}{|N_S(x)|}\sum_{y \in N_S(x)} \bigl(V(y) - V(x)\bigr)$
    \STATE $\text{new\_filtered} \leftarrow$ count of new filter-matching neighbors found
    \STATE Update stall counter: reset if $\text{new\_filtered} > 0$, else increment

    \STATE \textbf{--- Phase Logic ---}
    \IF{phase $= 1$}
        \IF{$\text{drift} < 0$}
            \STATE $\text{candidates} \leftarrow \{y \in N_S(x) : V(y) < V(x),\; y \text{ unexpanded}\}$, sorted by $V$
            \STATE Push top $K_f$ candidates onto frontier
        \ELSE
            \STATE $\text{phase} \leftarrow 2$; seed beam from $N(x) \cup \text{frontier}$; clear frontier
        \ENDIF
    \ELSIF{phase $= 2$}
        \STATE Add unseen neighbors of $x$ to beam; sort and prune to $B$
        \IF{$\text{drift} < 0$ \AND $\text{new\_filtered} > 0$}
            \STATE Rebuild frontier from beam's filtered unexpanded nodes
            \IF{frontier non-empty}
                \STATE $\text{phase} \leftarrow 1$; clear beam
            \ENDIF
        \ENDIF
    \ENDIF
\ENDWHILE

\RETURN results
\end{algorithmic}

\medskip
\paragraph{Phase 1 (Fiber Descent).} When $\text{drift}(x) < 0$, the filtered neighborhood slopes toward $q$. Only filtered, descending neighbors are queued (at most $K_f$ per step). This is cheap: every expansion targets a filter-matching, query-approaching point.

\paragraph{Phase 2 (Full-Graph Beam).} When $\text{drift}(x) \geq 0$, the fiber is flat, uphill, or absent. The walk falls back to standard beam search on the full graph with passive result collection. Three termination conditions: (a) beam converged, (b) best candidate's potential exceeds $k$-th best result $V_{(k)}$, (c) stall budget $T$ exhausted.

\paragraph{Phase Switching.} Phase 2 $\to$ Phase 1 re-entry requires $\text{drift} < 0$ \emph{and} $\text{new\_filtered} > 0$, ensuring the fiber is actively producing results, not merely theoretically present.

\subsection{Complexity} 

Let $R = |N(x)|$ denote the graph degree. Each expansion in either phase computes cosine similarity for unseen neighbors at cost $O(R \cdot d)$, which dominates. The drift computation uses cached potentials and costs $O(|N_S(x)|) \subseteq O(R)$---effectively free. The phases differ in frontier management:
\begin{itemize}
    \item \emph{Phase 1:} frontier insertion costs $O(K_f \log |\text{frontier}|)$ per step. Since the frontier contains only filtered descending neighbors, $|\text{frontier}| \leq K_f \cdot H_1$ where $H_1$ is the number of Phase~1 hops. With $K_f = 5$, this is negligible.
    \item \emph{Phase 2:} beam pruning costs $O(B \log B)$ per step, identical to standard beam search.
\end{itemize}
The total cost per walk is $O\bigl((H_1 + H_2) \cdot R \cdot d + H_2 \cdot B \log B\bigr)$, where $H_1$ and $H_2$ are the hops spent in each phase. When the fiber is dense, $H_1 \gg H_2$ and the walk avoids the $B \log B$ beam overhead entirely. When the fiber is sparse, Phase~2 dominates but the stall budget $T$ bounds $H_2$, ensuring the walk terminates and defers to anchor restarts rather than exhausting the expansion budget.

\paragraph{Anchor scoring cost.} With the inverted cluster index, each restart first retrieves the set of matching clusters $C_{\text{match}}$ in $O(|S|)$ via set intersection, then computes cosine similarity only for those clusters at cost $O(|C_{\text{match}}| \cdot d)$. Across $J$ restarts (with $|C_{\text{match}}|$ shrinking as processed clusters are excluded), the total anchor overhead per query is $O(|C_{\text{match}}| \cdot d)$---independent of the total number of clusters $K$. For selective filters ($|C_{\text{match}}| \ll K$), this is negligible compared to the walk cost.

\paragraph{Total query cost.} Filter membership is checked inline via metadata dict lookup at $O(|S|)$ per node, cached so each node is checked at most once. Over a walk visiting $H \cdot R$ nodes, this adds $O(H \cdot R \cdot |S|)$. Combining with the walk and anchor costs:
\[
    O\Bigl(\underbrace{(H_1 + H_2) \cdot R \cdot d}_{\text{similarity (dominates)}} + \underbrace{(H_1 + H_2) \cdot R \cdot |S|}_{\text{filter checking}} + \underbrace{|C_{\text{match}}| \cdot d}_{\text{anchor scoring}}\Bigr).
\]
Since $|S|$ is usually small compared to $d$, filter checking is negligible relative to similarity computation.

\paragraph{Scaling.} At large scale, the walk cost grows only through the graph degree $R$ and the number of hops needed for convergence---both are independent of $n$ for fixed graph construction parameters. The anchor scoring cost depends on $|C_{\text{match}}|$, which grows sublinearly with
$K = \sqrt{n}$ since selective filters match a diminishing fraction of clusters. Two strategies keep anchor scoring negligible: a hierarchical anchor structure (Section~\ref{sec:scale}) that reduces scoring to $O(n^{1/4} \cdot d)$ per query, or fixing $K$ to a moderate constant and relying on the walk to compensate for coarser cluster placement. In either case, the overall query cost remains dominated by the walk.

\section{Experimental Setup}\label{sec:experiments}

\paragraph{Dataset.} H\&M product embeddings: $n = 105{,}100$ vectors in $d = 2048$ dimensions with 24 categorical metadata fields (product group, colour, department, section, garment group, etc.) \cite{HM}. Ground truth: 10{,}000 test queries with conjunctive match filters and precomputed nearest filtered neighbors under cosine similarity. Filter selectivity ranges from ${<}\,0.01\%$ (e.g., ``Chem.\ cosmetics'') to ${\sim}21\%$ (e.g., ``Dusty Light'').

\paragraph{Our configurations.} $\alpha$-$k$NN graph with $k{=}64$ (mean degree ${\sim}128$),
$K{=}324$ clusters, jump budget $J{=}3$, seeds drawn from up to 5 clusters per restart ($C_{\text{max}} = 5$), and $n_s = 10$.
Beam search: $B{=}40$. Guided search: $B{=}2$ (Phase~2), $K_f{=}5$, stall budget $T{=}100$, max hops per walk 100. 
Source code is available at \cite{dang2026fibered_ann}.

\paragraph{Graph statistics.} The two graph variants exhibit different structural profiles on this dataset:

\begin{center}
\begin{tabular}{@{}lcccccc@{}}
\toprule
Graph & Total edges & Mean deg. & Min deg. & Max deg. & Memory \\
\midrule
$\alpha$-$k$NN & 8.7\,M & 82.8 & 19 & 128 & 34.8\,MB \\
HNSW base & 2.2\,M & 20.8 & 1 & 126 & 8.7\,MB \\
\bottomrule
\end{tabular}
\end{center}

The base layer HNSW is taken from FAISS with $M = 64$. From the table, the $\alpha$-$k$NN graph has $4\times$ the edges and $4\times$ the mean degree of the HNSW base graph. Higher connectivity benefits passive filtered collection: each expansion exposes more neighbors, increasing the probability of encountering a filter-matching point. The HNSW graph, by contrast, is aggressively pruned for efficient unfiltered greedy descent---its minimum degree of 1 means some nodes have a single edge, creating potential dead ends for filtered navigation. Despite this disadvantage, guided search on the HNSW graph still achieves high recall, confirming that the anchor atlas and drift-guided phase switching compensate for sparse graph structure.

\paragraph{HNSW baselines.} We use FAISS’s HNSW implementation as a baseline \cite{johnson2019billion} with $M{=}64$ (base-layer degree ${\sim}128$, matching our graph) and \texttt{efSearch}${=}400$:
\begin{itemize}
    \item \textbf{Post-filter:} retrieve $k {\times} 20$ unfiltered results, discard non-matching points.
    \item \textbf{Traversal-filter:} traversal-time filtering via \texttt{IDSelectorBatch} in \texttt{SearchParametersHNSW}. FAISS evaluates \texttt{sel->is\_member()} during candidate processing, collecting only filter-matching points as results while navigating the full graph.
\end{itemize}

\section{Results}\label{sec:results}

The table below compares our results with FAISS’s HNSW baseline methods for filtered search. Note that beam search uses the single-phase walk strategy described in
Section~\ref{sec:beam}, whereas guided search uses the drift-guided two-phase algorithm presented in Section~\ref{sec:guided}.

\begin{table}[H]
\centering
\caption{Recall@25 on 10{,}000 queries. Beam search uses $B{=}40$. Guided search uses $B{=}2$. HNSW uses \texttt{efSearch}${=}400$. Latency is measured on Apple M1. Our implementation uses Python with NumPy; FAISS uses optimized C++.}

\begin{tabular}{@{}lccccc@{}}

\toprule

Method & Recall@25 & $\geq 0.8$ & $= 1.0$ & Zero recall & Latency \\

\midrule

FAISS HNSW post-filter  & 0.313 & 24.1\% & 21.2\% & $44.8\%$ & 3.7\,ms \\

FAISS HNSW traversal-filter   & 0.489 & 38.8\% & \textbf{27.1}\% & $25.1\%$ & 3.4\,ms \\

Beam search, HNSW base ($B{=}40$) & 0.652 & 33.4\% & 21.2\% & 0.07\% & 9.6\,ms \\

Guided search, HNSW base ($B{=}2$) & 0.680 & 37.2\% & 10.5\% & 0.06\% & 19.0\,ms \\

Beam search, $\alpha$-$k$NN base ($B{=}40$) & 0.717 & 45.1\% & 21.7\% & 0.09\% & 14.3\,ms \\

Guided search, $\alpha$-$k$NN base ($B{=}2$) & \textbf{0.781} & \textbf{60.1\%} & 20.5\% & \textbf{0.05}\% & 31.9\,ms \\

\bottomrule

\end{tabular}
\end{table}

\begin{table}[H]
\centering
\caption{Walk statistics ($k{=}25$, 10{,}000 queries).}
\label{tab:walk-stats}
\begin{tabular}{@{}lcc@{}}
\toprule
& \textbf{Guided search $\alpha$-$k$NN ($B{=}2$)} & \textbf{Beam search $\alpha$-$k$NN ($B{=}40$)} \\
\midrule
Mean walks & 1.72 & 1.97 \\
Resolved in 1 walk & 60.6\% & 45.3\% \\
Mean hops & 151.4 & 84.4 \\
\midrule
\multicolumn{3}{@{}l}{\emph{Recall progression (queries needing $\geq j$ walks):}} \\
After walk 1 & 0.648 & 0.550 \\
After walk 2 & 0.543 & 0.508 \\
After walk 3 & 0.672 & 0.643 \\
After walk 4 & 0.804 & 0.795 \\
\bottomrule
\end{tabular}
\end{table}

Note that walk 1 recall is high because it includes the 60.6\% of queries that resolve in a single walk. Those might be easy queries with favorable fiber geometry. Queries reaching walk 2 and beyond are progressively harder, so the conditional recall dips before the restart mechanism recovers it.

\paragraph{Key observations.}
\begin{enumerate}
    \item Guided search with a $20\times$ smaller beam (B=2 vs.\ B=40) outperforms standard beam search on $\alpha$-$k$NN graph and on HNSW base graph, confirming that drift-directed fiber descent is more efficient than blind beam exploration for filtered collection.
    \item All of our methods outperform traditional HNSW recall. The anchor atlas bypasses the entry-point bottleneck that causes traditional HNSW to miss geometrically distant fibers.
    \item Half of the queries resolve in a single walk. For queries that require multiple walks, the anchor restart mechanism progressively recovers recall: the hardest queries (those reaching walk~4) recover from $\sim$0.54 after walk~2 to $\sim$0.80 after walk~4 (Table~\ref{tab:walk-stats}). 
\end{enumerate}

\paragraph{Why HNSW fails on selective filters.} When filtered nearest neighbors are geometrically distant from unfiltered nearest neighbors, both strategies degrade:
\begin{itemize}
    \item \emph{Post-filter:} the top-$N$ unfiltered results contain few or zero matching points. Example: a query whose overall nearest neighbors are garments ($\cos \approx 0.90$) but whose nearest shoes are at $\cos \approx 0.22$---no shoes appear in the top 500.
    \item \emph{Traversal-filter:} the HNSW walk starts from a hierarchical entry point and converges in the high-similarity region. With traversal-time filtering, it cannot bridge the graph distance to the distant fiber-present region---a \emph{topological cut}, exactly as predicted by the geometric framework.
\end{itemize}
Our approach bypasses this bottleneck: the anchor atlas seeds the walk directly in a fiber-present region near $q$, avoiding the entry-point funnel entirely.

\section{Stall Regime Analysis}\label{sec:analysis}

The geometric framework predicts three distinct failure modes for filtered graph search,
each characterized by local signals at the stall point. This section validates
those predictions empirically by classifying every walk termination across 1{,}000 queries
and examining how the distribution of stall regimes varies with filter selectivity.

\subsection{Stall Regimes}\label{sec:stall}

A \emph{stall point} $x^*$ is the last node expanded by a walk before termination---the node whose neighborhood was fully explored but failed to yield sufficient progress. Note that by the time the walk reaches $x^*$, it has already exited Phase~1 (which requires negative drift) and entered Phase~2. This means the drift at $x^*$ is non-negative, and it is unlikely that descent toward $q$ via filtered points remains available. The diagnostic question is: \textit{does the full
graph know a better direction that the fiber cannot access?} To answer this, we introduce the \emph{boundary-improving set} $B^-(x^*) = \{y \in N(x^*) \setminus X_S \mid V(y) < V(x^*)\}$, which contains neighbors outside the filter that would improve the potential.

Three regimes are possible. If the fiber is locally absent ($\rho_S$ low), the walk exhausted its budget without finding filtered neighbors---a \emph{topological cut} (Figure \ref{fig:top-cut})). If the fiber is present but slopes away from $q$ ($\rho_S$ high, $|B^-| > 0$), the descent directions pass through unfiltered territory---a \emph{geometric fold} (Figure \ref{fig:geometry-fold}). If no neighbor, filtered or otherwise, offers improvement ($\rho_S$ high, $|B^-| = 0$), the walk has reached a true local minimum---a \emph{genuine basin} (Figure \ref{fig:basin}).

\begin{center}
\begin{tabular}{lll}
\toprule
Condition & Interpretation & Regime \\
\midrule
$\rho_S(x^*)$ low & Fiber locally disconnected & Topological Cut \\
$\rho_S(x^*)$ high, $|B^-(x^*)| > 0$ & Fiber present, geometry misaligned & Geometric Fold \\
$\rho_S(x^*)$ high, $|B^-(x^*)| = 0$ & Genuine local basin & Genuine Basin \\
\bottomrule
\end{tabular}
\end{center}

\begin{figure}[H]
\centering
\begin{subfigure}[b]{0.3\textwidth}
    \centering
    \includegraphics[width=\textwidth]{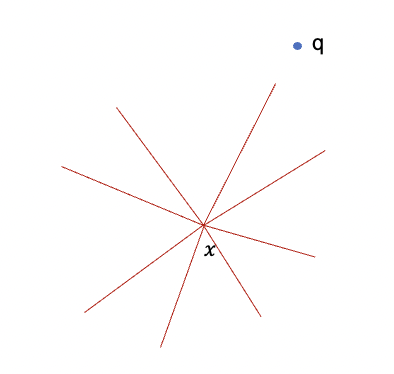}
    \caption{Topological cut}
    \label{fig:top-cut}
\end{subfigure}
\begin{subfigure}[b]{0.3\textwidth}
    \centering
    \includegraphics[width=\textwidth]{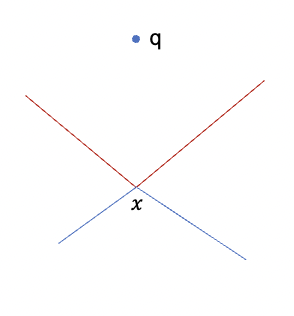}
    \caption{Geometric fold}
    \label{fig:geometry-fold}
\end{subfigure}
\begin{subfigure}[b]{0.3\textwidth}
    \centering
    \includegraphics[width=\textwidth]{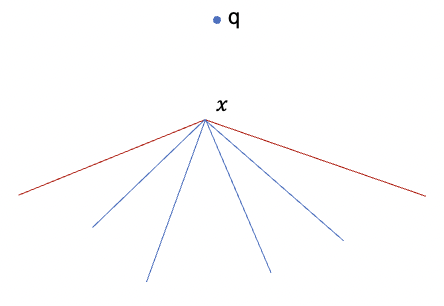}
    \caption{Genuine basin}
    \label{fig:basin}
\end{subfigure}
\caption{Stall regimes.}
\label{fig:fibers}
\end{figure}

\paragraph{All Failures Reduce to Restarting.}\label{sec:restart-insight}

The three stall regimes share a common resolution: the walk needs a better starting point in
a region where the fiber is present near $q$. Topological cuts need a region where filtered
points exist. Geometric folds need to approach $q$ from a different direction. Genuine basins
need a different basin of attraction. In all cases, the remedy is: \emph{jump to a new region
with metadata present near $q$ and restart.}

\subsection{Methodology}

A \textit{stall point} is classified as a \textit{topological cut} when $\rho_S(x^*) < \sigma / 2$,
where $\sigma = |X_S| / n$ is the global filter selectivity. This threshold
identifies nodes where the fiber is locally depleted relative to the global
baseline: if fewer than half the expected fraction of neighbors match the
filter, the fiber is considered locally disconnected. When
$\rho_S(x^*) \geq \sigma / 2$ and $|B^-(x^*)| > 0$, the stall is classified as a
\textit{geometric fold}; when $\rho_S(x^*) \geq \sigma / 2$ and $|B^-(x^*)| = 0$, it is
classified as a \textit{genuine basin}.

For this analysis, we increase the maximum hops per walk from 100 to 500  to allow the stall budget ($T{=}100$) to trigger independently of the hop limit and set $B=4$ to have slightly more exploration. All other parameters remain as in Section~\ref{sec:experiments}.
At each walk termination, we record the stall point $x^*$: the last expanded node before
the walk stopped. At $x^*$ we compute fiber density $\rho_S(x^*)$, drift, the
boundary-improving set $B^-(x^*)$, the potential $V(x^*)$, and the termination reason.
Each stall is classified into one of three regimes following the criteria in
Section~\ref{sec:stall}: topological cut ($\rho_S$ low), geometric fold ($\rho_S$ high,
$|B^-| > 0$), or genuine basin ($\rho_S$ high, $|B^-| = 0$). Queries are binned by
filter selectivity $|X_S|/n$ into five ranges from ${<}\,0.1\%$ to ${>}\,20\%$.

\subsection{Results by Selectivity Bin}

Table~\ref{tab:selectivity-bins} reports search performance and stall regime distribution
across selectivity bins.

\begin{table}[H]
\centering
\caption{Search performance and stall regime distribution by filter selectivity
(guided search, $B{=}4$, $k{=}25$, $\alpha$-$k$NN graph, 1000 queries).}
\label{tab:selectivity-bins}
\begin{tabular}{@{}lrrrrrrr@{}}
\toprule
\textbf{Selectivity} & \textbf{N} & \textbf{Recall} & \textbf{Hops} &
\textbf{Walks} & \textbf{TopCut} & \textbf{Fold} & \textbf{Basin} \\
\midrule
${<}\,0.1\%$  & 184 & 0.884 & 371.9 & 2.38 & 98.4\% &  1.6\% &  0.0\% \\
$0.1$--$1\%$  & 232 & 0.733 & 396.6 & 1.93 & 92.6\% &  7.4\% &  0.0\% \\
$1$--$5\%$    & 299 & 0.837 & 246.5 & 1.30 & 70.4\% & 29.1\% &  0.5\% \\
$5$--$20\%$   & 189 & 0.914 & 103.3 & 1.03 & 47.7\% & 48.7\% &  3.6\% \\
${>}\,20\%$   &  96 & 0.959 &  37.9 & 1.01 & 25.8\% & 53.6\% & 20.6\% \\
\bottomrule
\end{tabular}
\end{table}

The regime distribution shifts systematically with selectivity, confirming the framework's
predictions. At the rarest filters (${<}\,0.1\%$), nearly all stalls (98.4\%) are
topological cuts: the fiber is so sparse that the walk simply cannot find filtered neighbors
locally. As selectivity increases, topological cuts give way to geometric folds. At
${>}\,20\%$ selectivity, folds dominate (53.6\%) and genuine basins appear (20.6\%).

Notably, the lowest recall occurs not at the rarest filters but in the $0.1$--$1\%$ bin
(0.733 vs.\ 0.884 for ${<}\,0.1\%$). This is because recall is measured as a fraction of
the ground truth set. At extreme sparsity, the ground truth set itself is very small
(often fewer than $k{=}25$ matching points in the entire dataset), so finding even a modest
number of filtered points yields high fractional recall. The $0.1$--$1\%$ bin is the
hardest regime in absolute terms: the fiber is sparse enough to cause topological cuts
(92.6\%) but the ground truth set is large enough that missing points significantly
reduces recall.

The underlying mechanism is that geometric folds \emph{require the fiber to be present} in
order to be misaligned. When the fiber is absent ($\rho_S \approx 0$), the only possible
failure is topological: there is no fiber structure in the local neighborhood to be geometrically unfavorable.
As the fiber becomes denser, it exists at the stall point but may slope away from $q$---a
fold. At high density, the fiber is both present and well-connected, so stalls reflect
genuine local minima of the potential on the fiber rather than structural damage.

\subsection{Termination Reasons}

Table~\ref{tab:termination} shows how termination reasons shift with selectivity.

\begin{table}[H]
\centering
\caption{Walk termination reason distribution by filter selectivity.}
\label{tab:termination}
\begin{tabular}{@{}lrrr@{}}
\toprule
\textbf{Selectivity} & \textbf{Early stop} & \textbf{Stall budget} & \textbf{Max hops} \\
\midrule
${<}\,0.1\%$  &  0.2\% & 96.1\% & 3.7\% \\
$0.1$--$1\%$  & 17.2\% & 72.1\% & 10.5\% \\
$1$--$5\%$    & 54.4\% & 37.1\% & 8.2\% \\
$5$--$20\%$   & 90.8\% &  5.1\% & 4.1\% \\
${>}\,20\%$   & 97.9\% &  2.1\% & 0.0\% \\
\bottomrule
\end{tabular}
\end{table}

The pattern is consistent with the regime classification and now reveals three distinct
termination behaviors. At rare filters (${<}\,0.1\%$), the stall budget dominates (96.1\%):
the walk explores the full graph productively but encounters no filtered neighbors for 100 consecutive expansions. As selectivity increases, early stopping takes over: the walk collects enough high-quality filtered results that the $k$-th best potential gates further exploration, indicating healthy fiber descent.
Max hops termination is rare across all bins, confirming that the stall budget effectively
identifies fiber-sparse regions and triggers restarts before the expansion budget is exhausted.

\subsection{Stall Point Diagnostics}

Table~\ref{tab:stall-details} reports mean diagnostic values at the stall point $x^*$, aggregated by regime.

\begin{table}[H]
\centering
\caption{Mean diagnostic values at stall point $x^*$ by regime.}
\label{tab:stall-details}
\begin{tabular}{@{}lrrrrrr@{}}
\toprule
\textbf{Regime} & \textbf{Count} & $\boldsymbol{\rho_S}$ &
$\boldsymbol{|B^-|}$ & \textbf{Drift} & $\boldsymbol{V(x^*)}$ &
\textbf{Recall} \\
\midrule
Topological cut  & 1237 & 0.0024 & 20.5 & 0.0890 & 0.4046 & 0.762 \\
Geometric fold   &  300 & 0.1612 & 13.8 & 0.0716 & 0.3731 & 0.945 \\
Genuine basin    &   29 & 0.6689 &  0.0 & 0.1078 & 0.2407 & 0.992 \\
\bottomrule
\end{tabular}
\end{table}

The three regimes separate cleanly along every diagnostic axis:

\paragraph{Fiber density.}
Topological cuts occur at $\rho_S = 0.0024$---essentially zero. The fiber does not exist
locally. Geometric folds occur at $\rho_S = 0.161$: the fiber is present and navigable,
but misaligned. Genuine basins occur at $\rho_S = 0.669$: the fiber is dense and
well-connected.

\paragraph{Boundary-improving set.}
Topological cuts have the largest $|B^-|$ (20.5): many unfiltered neighbors would improve
the potential, confirming that the full graph has descent directions that the fiber cannot
access. Geometric folds also have positive $|B^-|$ (13.8): improving directions exist
outside the fiber, but unlike topological cuts, the fiber itself is present---it simply
does not align with those directions. Genuine basins have $|B^-| = 0$ exactly: no
neighbor, filtered or unfiltered, offers improvement. The walk has reached a true local
minimum.

\paragraph{Potential at stall.}
Topological cuts stall farthest from $q$ ($V = 0.405$): the walk cannot approach because
the fiber is absent in the region near $q$. Geometric folds stall closer ($V = 0.373$):
the walk approaches $q$ on the full graph but the fiber neighborhood does not follow.
Genuine basins stall closest ($V = 0.241$): the walk converges to the actual neighborhood
of $q$ on the fiber.

\paragraph{Recall.}
The recall progression across regimes ($0.762 \to 0.945 \to 0.992$) reflects
the severity of each failure mode. Topological cuts are the hardest to recover from
because the fiber must be found in a completely different region; anchor restarts address
this but cannot always recover full recall within the jump budget. Geometric folds are
milder: the fiber exists nearby, and a restart from a different direction typically
resolves the misalignment. Genuine basins require minimal recovery---the walk has
already found most of the nearest filtered neighbors.

\section{Conclusion}\label{sec:conclusion}

We presented a geometric framework for filtered approximate nearest neighbor search. Two local signals---fiber density, and drift---drive a two-phase search algorithm that begins with descent on filtered neighbors and falls back to full-graph exploration when the local geometry is unfavorable. When the walk stalls, a lightweight anchor atlas restarts the search in a fiber-present region near the query. These same signals classify search failures into three
regimes: topological cuts, geometric folds, and genuine basins. The empirical analysis confirmed that the three regimes separate cleanly and shift predictably with filter selectivity: rare filters fail topologically, common filters fail geometrically. On a real-world dataset with filter selectivities spanning three orders of magnitude, the method outperforms FAISS HNSW under both post-filtering and traversal-time filtering, with near-zero failure rate.

\paragraph{Limitations.}
Our experiments use only one real-world dataset (105{,}100 vectors, 24 categorical metadata fields).
Validation on larger-scale datasets and different domains would strengthen the generality claim. 
The implementation is in Python; a compiled implementation would be needed to evaluate latency competitively against production systems. 

\paragraph{Future directions.}
The geometric framework is not tied to any specific base graph, though our results suggest a preference for graphs with dense local structure---characterizing which graph families best support fiber descent is an open question. 
The anchor atlas currently handles conjunctive filters with single-value equality constraints; extending to set-valued predicates within a field is mechanically straightforward (Section \ref{sec:anchor}) but has not been evaluated. 
Range predicates would require additional per-cluster structures such as histograms or min/max bounds. The hierarchical anchor structure described in Section~\ref{sec:scale} offers a path for scaling to larger datasets but requires empirical validation. 
The stall classification threshold between low and high $\rho_S$ is currently heuristic; a principled, data-dependent threshold---perhaps derived from the global selectivity and graph degree---would make the diagnostic framework more robust.
Finally, we plan to benchmark against recent filtered search methods such as Filtered-DiskANN and ACORN on standard evaluation datasets.

\bibliographystyle{plain}
\bibliography{references}

@software{HM,
  author = {Qdrant},
  title = {ANN Filtered Retrieval Datasets},
  year = {2022},
  note = {\url{https://github.com/qdrant/ann-filtering-benchmark-datasets}},
}

@software{dang2026fibered_ann,
  author = {Dang, Thuong},
  title = {Fiber-Navigable Search: Source Code},
  year = {2026},
  note = {\url{https://github.com/thuongtuandang/fibered_ann}},
}

@book{lee2013smooth,
  title     = {Introduction to Smooth Manifolds},
  author    = {Lee, John M.},
  publisher = {Springer},
  year      = {2013},
  edition   = {2nd},
  series    = {Graduate Texts in Mathematics},
  volume    = {218},
  doi       = {10.1007/978-1-4419-9982-5}
}

@article{aitoamar2025rwalks,
  title   = {RWalks: Random Walks as Attribute Diffusers for Filtered Vector Search},
  author  = {Ait Aomar, Anas and Echihabi, Karima and Arnaboldi, Marco and Alagiannis, Ioannis and Hilloulin, Damien and Cherkaoui, Manal},
  journal = {Proceedings of the ACM on Management of Data},
  year    = {2025},
  doi     = {10.1145/3725349}
}

@article{cai2024ung,
  title   = {Navigating Labels and Vectors: A Unified Approach to Filtered Approximate Nearest Neighbor Search},
  author  = {Cai, Yuzheng and Shi, Jiayang and Chen, Yizhuo and Zheng, Weiguo},
  journal = {Proceedings of the ACM on Management of Data},
  year    = {2024},
  doi     = {10.1145/3698822}
}

@inproceedings{gollapudi2023filtereddiskann,
  title     = {Filtered-DiskANN: Graph Algorithms for Approximate Nearest Neighbor Search with Filters},
  author    = {Gollapudi, Siddharth and Karia, Neel and Sivashankar, Varun and Krishnaswamy, Ravishankar and Begwani, Nikit and Raz, Swapnil and Lin, Yiyong and Zhang, Yin and Mahapatro, Neelam and Srinivasan, Premkumar and Singh, Amit and Simhadri, Harsha Vardhan},
  booktitle = {Proceedings of the ACM Web Conference (WWW)},
  year      = {2023},
  doi       = {10.1145/3543507.3583552}
}

@article{iff2025benchmark,
  title   = {Benchmarking Filtered Approximate Nearest Neighbor Search Algorithms on Transformer-based Embedding Vectors},
  author  = {Iff, Patrick and Bruegger, Paul and Chrapek, Marcin and Kochergin, David and Besta, Maciej and Hoefler, Torsten},
  journal = {arXiv preprint arXiv:2507.21989},
  year    = {2025},
  url     = {https://arxiv.org/abs/2507.21989}
}

@article{johnson2019billion,
  title   = {Billion-scale similarity search with {GPUs}},
  author  = {Johnson, Jeff and Douze, Matthijs and J{\'e}gou, Herv{\'e}},
  journal = {IEEE Transactions on Big Data},
  year    = {2019},
  volume  = {7},
  number  = {3},
  pages   = {535--547}
}

@article{lin2025survey,
  title   = {Survey of Filtered Approximate Nearest Neighbor Search over the Vector-Scalar Hybrid Data},
  author  = {Lin, Yanjun and Zhang, Kai and He, Zhenying and Jing, Yinan and Wang, X. Sean},
  journal = {arXiv preprint arXiv:2505.06501},
  year    = {2025},
  url     = {https://arxiv.org/abs/2505.06501}
}

@article{malkov2020hnsw,
  title   = {Efficient and Robust Approximate Nearest Neighbor Search Using Hierarchical Navigable Small World Graphs},
  author  = {Malkov, Yu. A. and Yashunin, D. A.},
  journal = {IEEE Transactions on Pattern Analysis and Machine Intelligence},
  year    = {2020},
  volume  = {42},
  number  = {4},
  pages   = {824--836}
}

@article{patel2024acorn,
  title   = {ACORN: Performant and Predicate-Agnostic Search Over Vector Embeddings and Structured Data},
  author  = {Patel, Liana and Kraft, Peter and Guestrin, Carlos and Zaharia, Matei},
  journal = {Proceedings of the ACM on Management of Data},
  year    = {2024},
  volume  = {2},
  pages   = {1--27},
  doi     = {10.1145/3654923}
}

@inproceedings{subramanya2019diskann,
  title     = {DiskANN: Fast Accurate Billion-Point Nearest Neighbor Search on a Single Node},
  author    = {Subramanya, Suhas Jayaram and Devvrit, Fnu and Simhadri, Harsha Vardhan and Krishnaswamy, Ravishankar and Kadekodi, Rohan},
  booktitle = {Advances in Neural Information Processing Systems (NeurIPS)},
  year      = {2019}
}

@inproceedings{wang2023nhq,
  title     = {An Efficient and Robust Framework for Approximate Nearest Neighbor Search with Attribute Constraint},
  author    = {Wang, Mengzhao and Lv, Lingwei and Xu, Xiaoliang and Wang, Yuxiang and Yue, Qiang and Ni, Jiongkang},
  booktitle = {Advances in Neural Information Processing Systems (NeurIPS)},
  year      = {2023}
}

@inproceedings{wu2022hqann,
  title     = {HQANN: Efficient and Robust Similarity Search for Hybrid Queries with Structured and Unstructured Constraints},
  author    = {Wu, Wei and He, Junlin and Qiao, Yu and Fu, Guoheng and Liu, Li and Yu, Jin},
  booktitle = {Proceedings of the 31st ACM International Conference on Information and Knowledge Management (CIKM)},
  year      = {2022},
  doi       = {10.1145/3511808.3557610}
}

@article{wang2025wow,
  title   = {WoW: A Window-to-Window Incremental Index for Range-Filtering Approximate Nearest Neighbor Search},
  author  = {Wang, Ziqi and Zhang, Jingzhe and Hu, Wei},
  journal = {Proceedings of the ACM on Management of Data},
  year    = {2025},
  doi     = {10.1145/3769843}
}

@article{zhao2022airship,
  title   = {Constrained Approximate Similarity Search on Proximity Graph},
  author  = {Zhao, Weijie and Tan, Shulong and Li, Ping},
  journal = {arXiv preprint arXiv:2210.14958},
  year    = {2022},
  url     = {https://arxiv.org/abs/2210.14958}
}

@article{zuo2024serf,
  title   = {SeRF: Segment Graph for Range-Filtering Approximate Nearest Neighbor Search},
  author  = {Zuo, Chaoji and Qiao, Miao and Zhou, Wenchao and Li, Feifei and Deng, Dong},
  journal = {Proceedings of the ACM on Management of Data},
  year    = {2024},
  doi     = {10.1145/3639324}
}

\end{document}